\documentclass[journal]{IEEEtran}

\usepackage{amsmath,amsfonts}
\allowdisplaybreaks[4]
\usepackage{amsthm}
\usepackage{algorithmic}
\usepackage{algorithm}
\usepackage{array}
\usepackage{subfigure}
\usepackage{textcomp}
\usepackage{stfloats}
\usepackage{url}
\usepackage{verbatim}
\usepackage{graphicx}
\usepackage{cite}
\usepackage{balance}
\usepackage{enumerate}
\usepackage{xcolor}
\newtheorem{lemma}{Lemma}

\begin{document}
\bibliographystyle{IEEEtran}

\title{STAR-RIS-aided NOMA for Secured xURLLC}

\author{Lulu Song, Di Zhang,~\IEEEmembership{Senior Member,~IEEE}, Shaobo Jia,~\IEEEmembership{Member,~IEEE}, Pengcheng Zhu,~\IEEEmembership{Member,~IEEE}, Yonghui Li,~\IEEEmembership{Fellow,~IEEE}
  \thanks{Copyright (c) 20xx IEEE. This is a pre-print version of article accepted by IEEE Transactions on Vehicular Technology. Personal use of this material is permitted. However, permission to use this material for any other purposes must be obtained from the IEEE by sending a request to pubs-permissions@ieee.org.}
  \thanks{This work was supported in part by the National Natural Science Foundation of China under Grant U22A2001, 62301502, and the Henan Natural Science Foundation for Excellent Young Scholar under Grant 242300421169. (Corresponding author: Di Zhang).}
  \thanks{Lulu Song, Di Zhang and Shaobo Jia are with the School of Electrical and Information Engineering, Zhengzhou University, and  also with the School of Electrical and Information Engineering, the Henan International Joint Laboratory of Intelligent Health Information System, the National Telemedicine center, the National Engineering Laboratory for Internet Medical Systems and Applications, Zhengzhou University, 450001, China (E-mail: lulu\_song@gs.zzu.edu.cn, dr.di.zhang@ieee.org, ieshaobojia@zzu.edu.cn).} 
 
  \thanks{Pengcheng Zhu is with the National Mobile Communications Research Laboratory, Southeast University, Nanjing 210096, China (E-mail: p.zhu@seu.edu.cn).}
  \thanks{Yonghui Li is with the School of Electrical and Information Engineering, University of Sydney, Sydney, NSW 2006, Australia (E-mail: yonghui.li@sydney.edu.au).}
}



\maketitle

\begin{abstract}
Short packet-based advanced Internet of things (A-IoT) calls for not only the next generation of ultra-reliable low-latency communications (xURLLC) but also highly secured communications. In this paper, we aim to address this objective by developing a  non-orthogonal multiple access (NOMA) system with untrusted user. There exist two key problems: The confidential/private message for the far user will be exposed to the untrusted near user with successful SIC; The restrictive trade-off among reliability, security and latency poses a great challenge in achieving secured xURLLC. In order to solve these issues, we introduce simultaneous transmitting and reflecting reconfigurable intelligent surface (STAR-RIS), which provides additional degree of freedom to enable a secure and fair decoding order and achieve a desired trade-off among reliability, security and latency. To fully reveal the trade-off among reliability, security and latency, we characterize the reliability and security via decoding error probabilities. A leakage probability minimization problem is modeled to optimize the passive beamforming, power allocation and blocklength subject to secure successive interference cancellation (SIC) order, reliability and latency constraints. To solve this complex problem, we explore its intrinsic properties and propose an algorithm based on majorization minimization (MM) and alternative optimization (AO). Simulation results demonstrate the validness of our study in this paper.
\end{abstract}


\IEEEpeerreviewmaketitle

\section{Introduction}
\IEEEPARstart{S}{ix} generation (6G) is anticipated to boost the proliferation of mission-critical Internet of Things (IoT) applications like industrial IoT and emergency rescue, which will be featured by short packet communications (SPC)\cite{ref1}. Besides stricter requirements on connection density (${10^6}$-${10^8}$ devices/${\rm{k}}{{\rm{m}}^2}$), reliability (${10^{ - 5}}$-${10^{ - 7}}$) and latency (0.1-1ms), security and privacy issues arising from ubiquitous connectivity are emerging as the new focuses \cite{ref2}. Moreover, the reliability and security performance of SPC can be significantly degraded due to the use of finite blocklength codes \cite{ref3,ref4}. Advanced technologies to support next-generation ultra-reliable low-latency communications (xURLLC) and secure transmissions in SPC-based IoT applications are thus emerging as a critical issue.

In literature, non-orthogonal multiple access (NOMA)-SPC has emerged as a promising technology to provide enhanced spectrum efficiency \cite{ref5}, connectivity and fairness \cite{ref7}, and reduced latency \cite{ref6}. Reliability performance of NOMA-SPC can be further improved by sparse vector coding-based multi-carrier technologies\cite{ref8}. On the other hand, alongside xURLLC, information security and privacy protection are surfacing as pivotal concerns in the evolution towards 6G. As a attractive complement to upper-layer encryption, physical layer security exploits the physical security attributes and has the advantages of lower complexity and higher flexibility\cite{ref91}.

In light of this, some works have been done to investigate the security in NOMA-SPC system. In \cite{ref9}, the secure energy efficiency of an uplink NOMA-SPC network in presence of an eavesdropper was maximized. In \cite{ref10}, the security performance of a downlink NOMA-SPC system with diversified requirements users was analyzed. However, most works on NOMA-SPC focused on external eavesdropping, while overlooking the potential risk of information leakage posed by untrusted internal user, especially in dynamic/heterogeneous networks. To ensure both spectrum efficiency and confidentiality/privacy, it is reasonable to consider a hostile but realistic NOMA scenario where some users want to safeguard their confidential/private information and treat other users as untrusted internal eavesdroppers. Although the untrusted NOMA has been considered in \cite{ref11,ref12}, it only assumed that the untrusted user was farther away from the source. The even worse scenario with untrusted near user is left out, where the confidential/private message for the far user can be overheard by the untrusted near user with successful SIC.

Besides, the interplay among reliability, security and transmission delay of SPC is important in meeting diverse and customized requirements of secured xURLLC. Focusing on single transmission with specific packet size, the authors in \cite{ref13} demonstrated that the the trade-off between reliability and security cannot be sufficiently represented by the metrics with given reliability/security constraint. Therefore, the leakage-failure probability was proposed to address the trade-off in the classic three-node setup. However, the interplay among reliability, security and transmission delay in NOMA-SPC has not yet been investigated. Moreover, due to the complexity and uncontrollability of the propagation environment, the interplay may present a formidable challenge in concurrently achieving satisfactory performances across all three aspects.

Fortunately, the electromagnetic environment can be proactively adjusted with very low consumption by reconfigurable intelligent surface (RIS), a prospective technology in 6G. When used in SPC, it can substantially enhance the performance such as achievable secrecy rate and block error rate \cite{ref14,ref15}. The benefits of using RIS in untrusted NOMA-SPC are two folds: By constructively/destructively reconfiguring the channel between the source and the secure/untrusted user, the effective channel gain of the secure user can be ensured to be relatively strong, and thus a secure and fair SIC order can be guaranteed; The ability to differentiate the channel gains provides additional degrees of design freedom, and thus the security performance can be enhanced while concurrently achieving better performance in both reliability and transmission delay. In other words, a more desired trade-off among reliability, security and latency can be achieved.

Motivated by these concerns, we investigate the secured xURLLC in NOMA-SPC system with untrusted near user. To reveal the trade-off among reliability, security and transmission delay of single transmission with specific packet size, decoding error probabilities are adopted to characterize the reliability and security performance \cite{ref13,ref16}. To enhance security while fulfilling stricter reliability and latency requirements, a novel RIS named simultaneous transmitting and reflecting RIS (STAR-RIS) \cite{ref17} is introduced to ensure the secure SIC order and provide higher freedom so that a more desired trade-off among reliability, security and latency can be achieved. Then the leakage probability is minimized under reliability and latency constraints. Finally, we conduct simulations to evaluate the proposed scheme in terms of the reliability-security-latency and convergence performance.

\section{System Model}
We consider a STAR-RIS-aided downlink NOMA-SPC system with untrusted user in the scenario of mission-critical IoT. As shown in Fig. \ref{fig1}, a single-antenna source (S) transmits messages to a pair of single-antenna NOMA users ${U_i},i {\in}\left\{ {c,s} \right\}$ with the assistance of a STAR-RIS (R) equipped with $N{=}{N_v}{N_h}$ elements, where ${N_v}$ and ${N_h}$ are the numbers of the elements along rows and columns. The confidential/private signal $s_s$ for user $U_s$ and the signal $s_c$ for user $U_c$ are superimposed. To prevent information leakage, We treat $U_c$ as untrusted user. Consider single transmission, where the confidential/private information of $d_s$ bits intended for $U_s$ and the information of $d_c$ bits intended for $U_c$ are transmitted within a shared blocklength of $m$ channel uses.
\begin{figure}[!t]
  \centering
  \includegraphics[width=1.6in]{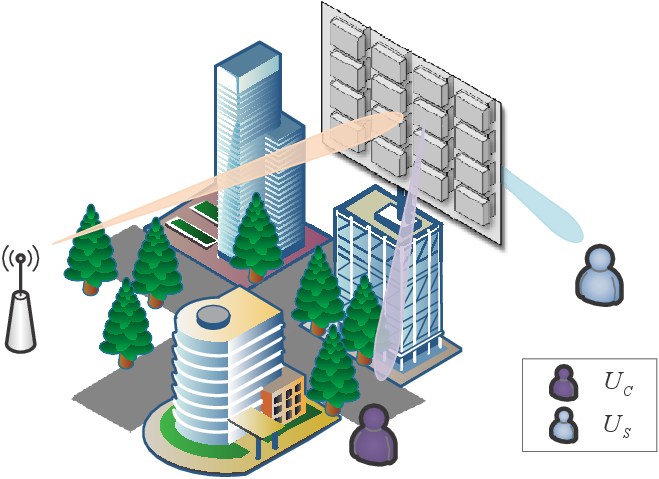}
  \caption{STAR-RIS-aided downlink NOMA}
  \label{fig1}
  \vspace{-0.3cm}
\end{figure}
Specifically, the superimposed transmitted signal is
\begin{equation}
  \label{1}
  \begin{split}
    x = \sqrt {{a_c}P} {s_c} + \sqrt {{a_s}P} {s_s},
  \end{split}
\end{equation}
where $P$ is the transmit power at S, ${a_i},i {\in} \left\{ {c,s} \right\}$ is the power allocation factor for ${U_i}$. Without loss of generality, we assume that ${U_c}$ is located in the reflection area and ${U_s}$ in the transmission area. The 3D coordinates of S, R, ${U_c}$ and ${U_s}$ are ${{\bf{c}}_j} {=} {\left[ {c_j^x,c_j^y,c_j^z} \right]^T},j {\in} \left\{ {S,R,{U_c},{U_s}} \right\}$. The direct links from S to ${U_i}$ are omitted here due to unfavorable obstacles. Then the received signals reflected/transmitted by STAR-RIS at ${U_i}$ can be given as
\begin{equation}
  \label{2}
  \begin{split}
    {y_i} = {\bf{v}}_i{{\bf{h}}_i}x + {n_i},
  \end{split}
\end{equation}
where ${{\bf{v}}_i} {=} {\left[ {{\mu _{i,1}}{e^{j{\theta _{i,1}}}}, \ldots ,{\mu _{i,N}}{e^{j{\theta _{i,N}}}}} \right]}$ characterizes the reflection/transmission property of STAR-RIS, ${\mu _{i,n}}{\in} \left[ {0,1} \right]$ and ${\theta _{i,n}} {\in} \left[ {0,2\pi } \right)$ are the reflection/transmission amplitude and phase shift of the $n$th element, respectively. Energy splitting protocol is adopted here due to its higher flexibility than mode switching and time switching protocol. According to the law of energy conservation, ${\mu _{i,n}}$ satisfies $\mu _{c,n}^2 {+} \mu _{s,n}^2 {\le} 1$ \cite{ref18,ref19}. In addition, ${{\bf{h}}_i} {=} \text{diag}\left( {{{\bf{h}}_{Ri}}} \right){{\bf{h}}_{SR}}$, where ${{\bf{h}}_{SR}}, {{\bf{h}}_{Ri}} {\in} {\mathbb{C}^{N \times 1}}$ represent the channels between S→R and R→${U_i}$. The additive white Gaussian noise at ${U_i}$ is denoted by ${n_i} {\sim} {\cal C}{\cal N}(0,\sigma _i^2)$. Under the Rician channel model, we have
\begin{equation}
  \label{3}
  \begin{split}
    {{\bf{h}}_{SR}} = \sqrt {\rho d_{SR}^{ - {\alpha _1}}} \Big( {\sqrt {\frac{K}{{1 + K}}} {\bf{h}}_{SR}^{LoS} + \sqrt {\frac{1}{{1 + K}}} {\bf{h}}_{SR}^{NLoS}} \Big),
  \end{split}
\end{equation}
\begin{equation}
  \label{4}
  \begin{split}
    {{\bf{h}}_{Ri}} = \sqrt {\rho d_{Ri}^{ - {\alpha _2}}} \Big( {\sqrt {\frac{K}{{1 + K}}} {\bf{h}}_{Ri}^{LoS} + \sqrt {\frac{1}{{1 + K}}} {\bf{h}}_{Ri}^{NLoS}} \Big),
  \end{split}
\end{equation}
where ${d_{SR}}$ and ${d_{Ri}}$ are the distance between two nodes, ${\alpha _1}$ and ${\alpha _2}$ are the path loss exponents, $\rho$ and $K$ denote the path loss at unit distance and the Rician factor. The line-of-sight components are defined as ${\bf{h}}_{SR}^{LoS} = {\bf{a}}\left( {{\phi _{SR}},{\varphi _{SR}}} \right)$ and ${\bf{h}}_{Ri}^{LoS} = {\bf{a}}\left( {{\phi _{Ri}},{\varphi _{Ri}}} \right)$, where ${\bf{a}}(\phi ,\varphi ) = \left[ {1, \ldots ,{e^{ - j\frac{{2\pi l}}{\lambda }\left( {{N_v} - 1} \right)\sin \phi \cos \varphi }}} \right]^T \otimes \left[ 1, \ldots,{e^{ - j\frac{{2\pi l}}{\lambda }\left( {{N_h} - 1} \right)\sin \phi \sin \varphi }} \right]^T$ is the array response vector, $l$ and $\lambda$ denote the inter-element distance and the signal wavelength, $\phi $ and $\varphi $ denote the elevation angle-of-arrival (AoA)/angle-of-departure (AoD) and the azimuth AoA/AoD, i.e., ${\phi _{SR}} = \arcsin \big( {\frac{{c_S^z - c_R^z}}{{\left\| {{{\bf{c}}_S} - {{\bf{c}}_R}} \right\|}}} \big)$, ${\phi _{Ri}} = \arcsin \big( {\frac{{c_{{U_i}}^z - c_R^z}}{{\left\| {{{\bf{c}}_{{U_i}}} - {{\bf{c}}_R}} \right\|}}} \big)$, ${\varphi _{SR}} = \arccos \big( {\frac{{c_S^x - c_R^x}}{{{{\left\| {{{\bf{c}}_S} - {{\bf{c}}_R}} \right\|}_{1:2}}}}} \big)$, ${\varphi _{Ri}} = \arccos \big( {\frac{{c_{{U_i}}^x - c_R^x}}{{{{\left\| {{{\bf{c}}_{{U_i}}} - {{\bf{c}}_R}} \right\|}_{1:2}}}}} \big)$, ${\bf{h}}_{SR}^{NLoS}$ and ${\bf{h}}_{Ri}^{NLoS}$ denote the non-line-of-sight components with circular symmetric complex Gaussian random variables.

\section{Reliability-Security-Latency Performance}
Prior study in \cite{ref13} indicates that a trade-off between reliability and security can be represented by the decoding error probabilities. Specifically, the unreliability is measured by the probability of the event that ${U_i}$ decodes ${s_i}$ incorrectly, i.e., the error probability $\varepsilon_{ii}$. The insecurity is measured by the probability of the event that ${U_c}$ decodes ${s_s}$ successfully, i.e., the leakage probability ${\delta _{sc}}$. We will follow this to analyze the reliability-security-latency performance.

For the sake of security, ${s_c}$ needs to be first decoded. Otherwise, the leakage probability will be high because the reliability at ${U_c}$ is on the premise of successful decoding of ${s_s}$. Thanks to the channel reconstruction capability of STAR-RIS, the effective channel gain for ${U_s}$ can be enhanced to be stronger, and thus the secure SIC order can be guaranteed regardless of whether ${U_s}$ is near or far from the source.

For a given packet size $d$, the block error probability is \cite{ref3}
\begin{equation}
  \label{5}
  \begin{split}
    \varepsilon \left( {\gamma ,m, d} \right) = Q\Big( {\sqrt {\frac{m}{{V\left( \gamma  \right)}}} \big( {C\left( \gamma  \right) - \frac{d}{m}} \big)\ln 2} \Big),
  \end{split}
\end{equation}
where $\gamma$, $C\left( \gamma  \right) {=} {\log _2}\left( {1 + \gamma } \right)$, $V\left( \gamma  \right) {=} 1 {-} {\left( {1 + \gamma } \right)^{-2}}$ and $Q\left( x \right) {=} \int_x^\infty  {\frac{1}{{\sqrt {2\pi } }}} {e^{ - {t^2}/2}}dt$ are the SNR, Shannon’s capacity, the channel dispersion and Gaussian Q function, respectively. 

In the first step of SIC, the signal-to-noise-and interference ratio (SINR) to decode ${s_c}$ at ${U_i}$ is 
\begin{equation}
  \label{6}
  \begin{split}
    {\gamma _{ci}} = \frac{{{a_c}P{{\left| {{\bf{v}}_i^H{{\bf{h}}_i}} \right|}^2}}}{{{a_s}P{{\left| {{\bf{v}}_i^H{{\bf{h}}_i}} \right|}^2} + \sigma _i^2}},i \in \left\{ {c,s} \right\}.
  \end{split}
\end{equation}
Then the decoding error probability of ${s_c}$ at ${U_i}$ can be obtained by ${\varepsilon _{ci}} = \varepsilon \left( {{\gamma _{ci}},m,d_c} \right)$. If ${s_c}$ is canceled successfully, the SINR to decode ${s_s}$ at ${U_i}$ can be denoted by
\begin{equation}
  \label{7}
  \begin{split}
    {\gamma _{si}} = \frac{{{a_s}P{{\left| {{\bf{v}}_i^H{{\bf{h}}_i}} \right|}^2}}}{{\sigma _i^2}},
  \end{split}
\end{equation}
and the conditional decoding error probability of ${s_s}$ at ${U_i}$ will be ${\bar \varepsilon _{si}} = \varepsilon \left( {{\gamma _{si}},m,d_s} \right)$. The effective decoding error probability of ${s_s}$ at ${U_i}$ can thus be given as
\begin{equation}
  \label{8}
  \begin{split}
    {\varepsilon _{si}} = \left( {1 - {\varepsilon _{ci}}} \right){\bar \varepsilon _{si}} + {\varepsilon _{ci}}.
  \end{split}
\end{equation}
Therefore, the unreliability can be characterized by the probability of the event that $U_c$ decodes $s_c$ incorrectly in the first SIC step, i.e., the error probability ${\varepsilon _{cc}}$, and the probability of the event that $U_s$ decodes $s_c$ incorrectly in the first SIC step or decodes $s_s$ incorrectly in the second step, i.e., the error probability ${\varepsilon _{ss}}$. The insecurity can be characterized by the probability of the event that $U_c$ decodes $s_c$ in the first SIC step and $s_s$ in the second step correctly, i.e., the leakage probability ${\delta _{sc}} = 1 - {\varepsilon _{sc}}$. Besides, the requirement on latency imposes constraint on the transmission delay per transmission, which can be defined as
\begin{equation}
  \label{9}
  \begin{split}
    {T} = \frac{m}{B},
  \end{split}
\end{equation}
where $B$ is the bandwidth. 

It is evident that there exists interplay among ${\varepsilon _{cc}}$, ${\varepsilon _{ss}}$, ${\delta _{sc}}$ and $T$. Specifically, since $\frac{{\partial {\varepsilon _{si}}}}{{\partial {\varepsilon _{ci}}}} = 1 - {\bar \varepsilon _{si}} \ge 0$, $\frac{{\partial {\varepsilon _{si}}}}{{\partial {{\bar \varepsilon }_{si}}}} = 1 - {\varepsilon _{ci}} \ge 0$, and $\frac{{\partial \varepsilon }}{{\partial m }} \le 0$ holds for both ${\varepsilon _{ci}}$ and ${\bar \varepsilon _{si}}$, we have $\frac{{\partial {\varepsilon _{cc}}}}{{\partial m}} \le 0$, $\frac{{\partial {\varepsilon _{ss}}}}{{\partial m}} = \frac{{\partial {\varepsilon _{ss}}}}{{\partial {\varepsilon _{cs}}}}\frac{{\partial {\varepsilon _{cs}}}}{{\partial m}} + \frac{{\partial {\varepsilon _{ss}}}}{{\partial {{\bar \varepsilon }_{ss}}}}\frac{{\partial {{\bar \varepsilon }_{ss}}}}{{\partial m}} \le 0$, $\frac{{\partial {\delta _{sc}}}}{{\partial m}} = -\frac{{\partial {\varepsilon _{sc}}}}{{\partial {\varepsilon _{cc}}}}\frac{{\partial {\varepsilon _{cc}}}}{{\partial m}} - \frac{{\partial {\varepsilon _{sc}}}}{{\partial {{\bar \varepsilon }_{sc}}}}\frac{{\partial {{\bar \varepsilon }_{sc}}}}{{\partial m}} \ge 0$ and $\frac{{\partial {T}}}{{\partial m}} = \frac{1}{B} \ge 0$. Therefore, ${\varepsilon _{cc}}$, ${\varepsilon _{ss}}$ decrease monotonically while ${\delta _{sc}}$, $T$ increase monotonically with respect to $m$.

With respect to ${a_c}$ and ${a_s}$, and since $\frac{{\partial \varepsilon }}{{\partial \gamma }} \le 0$ holds for both ${\varepsilon _{ci}}$ and ${\bar \varepsilon _{si}}$, we will have $\frac{{\partial {\varepsilon _{ci}}}}{{\partial a_c}} = \frac{{\partial {\varepsilon _{ci}}}}{{\partial {\gamma _{ci}}}}\frac{{\partial {\gamma _{ci}}}}{{\partial a_c}} \le 0$, $\frac{{\partial {\varepsilon _{ci}}}}{{\partial a_s}} = \frac{{\partial {\varepsilon _{ci}}}}{{\partial {\gamma _{ci}}}}\frac{{\partial {\gamma _{ci}}}}{{\partial a_s}} \ge 0$, $\frac{{\partial {\bar\varepsilon _{si}}}}{{\partial a_s}} = \frac{{\partial {\bar\varepsilon _{si}}}}{{\partial {\gamma _{si}}}}\frac{{\partial {\gamma _{si}}}}{{\partial a_s}} \le 0$. Therefore, allocating more power to $s_c$ and less power to $s_s$ results in lower $\varepsilon _{ci}$ and higher $\bar\varepsilon _{si}$. Even though a $\bar\varepsilon _{si}$ at a very low level facilitates a small $\delta _{sc}$, however, it leads to an unacceptable increase in $\varepsilon _{ss}$ at the same time.

Accordingly, it can be seen that there exists trade-off among ${\varepsilon _{cc}}$, ${\varepsilon _{ss}}$, ${\delta _{sc}}$ and $T$. For a clear illustration of the trade-off without STAR-RIS design, an explicit presentation of $({\varepsilon _{cc}}, {\varepsilon _{ss}}, {\delta _{sc}}, T)$ corresponding to different $(m, a_c)$ (assume $a_c+a_s=1$) is depicted in the form of 4D plot in fig. 2, where the fourth dimension is marked with color. In this work we aim at achieving satisfactory ${\varepsilon _{cc}}$, ${\varepsilon _{ss}}$, ${\delta _{sc}}$ and $T$ to fulfill the requirements of secured xURLLC. However, it can be seen clearly from fig. 2 that without STAR-RIS design, satisfactory ${\varepsilon _{cc}}$, ${\varepsilon _{ss}}$, ${\delta _{sc}}$ and $T$ cannot be achieved concurrently.

The analysis above indicates the significance and necessity of the additional degrees of design freedom provided by STAR-RIS to achieve a more desired trade-off. In the following, we aim at achieving higher security while fulfilling stringent reliability and latency requirements with the assistance of STAR-RIS through joint optimization.

\begin{figure}
    \centering
    \includegraphics[width=0.5\linewidth]{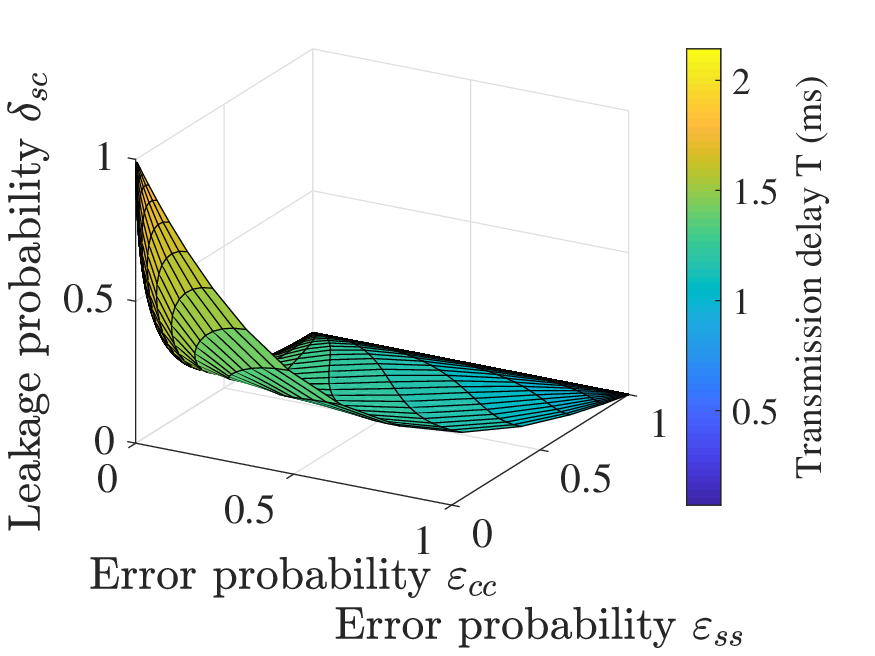}
    \caption{${\varepsilon _{cc}}$, ${\varepsilon _{ss}}$, ${\delta _{sc}}$ and $T$ without STAR-RIS design}
    \label{fig2}
    \vspace{-0.3cm}
\end{figure}

\section{Problem formulation and reformulation}
To maximize security while fulfilling customized latency and reliability requirements, an optimization problem is formulated in this section. Specifically, the passive beamforming, the blocklength and the power allocation are jointly optimized to minimize the leakage probability subject to constraints of reliability, transmission delay, SIC order, transmit power and STAR-RIS elements energy. Mathematically, the optimization problem can be given as
\begin{subequations}
  \begin{align}
     & {\rm{P1}}: \mathop {\min }\limits_{{{\bf{v}}_i},m,{a_i}} {\delta _{sc}} \label{10a} \\ 
     & \mbox{s.t.}\ {\varepsilon _{cc}} \le {\varepsilon _c},\label{10b} \\
     & \quad\ {\varepsilon _{ss}} \le {\varepsilon _s},\label{10c} \\
     & \quad\ {a_c} > {a_s},\label{10d}   \\
     & \quad\ {a_c} + {a_s} \le 1,\label{10e}\\
     & \quad\ {\left| {{{\bf{v}}_{c,n}}} \right|^2} + {\left| {{{\bf{v}}_{s,n}}} \right|^2} \le 1,\forall n \in \left\{ {1, \ldots ,N} \right\}, \label{10f}\\
     & \quad\ T \le {T_{max}},\ m \in \mathbb{N}^+,  \label{10g}
  \end{align}
\end{subequations}
where (\ref{10b}) and (\ref{10c}) ensure that the error probability at ${U_i}$ is no greater than the threshold ${\varepsilon _i}$. Constraints (\ref{10d}) and (\ref{10e}) guarantee the secure SIC order and a total transmit power within $P$, (\ref{10f}) is the energy-splitting constraint for the two modes of STAR-RIS elements, ${T_{max}}$ in constraint (\ref{10g}) represents the maximum tolerable transmission delay.

Note that P1 is intractable due to the complex expressions. Besides, the widely used approximation methods \cite{ref16} are incapable due to their inaccuracy. To solve it, we reformulate P1 by exploring its intrinsic properties.

\begin{lemma}
    \label{lemma2}
    Constraints (\ref{10b}) and (\ref{10e}) hold with equality at the optimal solution.
\end{lemma}

\begin{proof}
    See Appendix A.
\end{proof}

With Lemma 2 in hand, the objective of P1 can be transformed into maximizing ${\bar \varepsilon _{sc}}$. The "$\le$" in (\ref{10b}) and (\ref{10e}) can be replaced with "$=$". For constraint ${\varepsilon _{cc}} = {\varepsilon _c}$, the “$=$” can be relaxed to “$\le$” and the tightness can be proved similarly to lemma 2. Then recalling that $Q(x)$ decreases monotonically with $x$, P2 can be reformulated as
\begin{subequations}
  \begin{align}
     & {\rm{P2}}: \mathop {\min }\limits_{{{\bf{v}}_i},m,{a_c}} \sqrt {\frac{m}{{V\left( {{\gamma _{sc}}} \right)}}} \Big( {C\left( {{\gamma _{sc}}} \right) - \frac{d}{m}} \Big) \label{11a} \\
     & \mbox{s.t.} \quad {a_s} = 1 - {a_c},\label{11b} \\
     & \quad \quad {a_c}>0.5, \label{11c} \\
     & \quad \quad (\text{10b}),(\text{10c}),(\text{10f}),(\text{10g}).\label{11d} 
  \end{align}
\end{subequations}

P2 is nonconvex due to the discrete integer $m$ and the coupling between ${{\bf{v}}_i}$, $m$ and ${a_c}$. Hence, P2 is a highly nonconvex mixed integer non-linear program (MINLP), which is NP-hard and computationally intractable. So far, there is no known polynomial-time algorithm to obtain the globally optimal solution for MINLP \cite{ref20}. Considering the delay-sensitive requirement, an effective approximation algorithm is introduced in the next section.

\section{Optimization Algorithm}
In this section, we decouple P3 into two subproblems and propose an effective algorithm to solve them.

\subsubsection{Optimizing ${{\bf{v}}_i}$, $m$ with fixed ${a_c}$} Due to the complexity of MINLP and the delay-sensitive demand, the constraint on the blocklength $m$ is relaxed from $m \in \mathbb{N}^+$ to continuous constraint $m\ge0$ for a quicker resolution. Recall that ${\bar \varepsilon _{sc}}$ decreases monotonically with ${\gamma _{sc}}$, ${\varepsilon _{cc}}$ decreases monotonically with ${\gamma _{cc}}$, ${\varepsilon _{ss}}$ decreases monotonically with ${\gamma _{cs}}$ and ${\gamma _{ss}}$. Besides, both ${\gamma _{ci}}$ and ${\gamma _{si}}$ decrease monotonically with $\left| {{\bf{v}}_i^H{{\bf{h}}_i}} \right|$. By introducing auxiliary variables ${\alpha _{ci}}$, ${\alpha _{si}}$, ${\beta _i}$, the subproblem can be rewritten as
\begin{subequations}
  \begin{align}
     & {\rm{P3}}: \mathop {\min }\limits_{\scriptstyle{{\bf{v}}_i},m,\hfill\atop
\scriptstyle{\alpha _{ci}},{\alpha _{si}},{\beta _i}\hfill} \sqrt {\frac{m}{{V\left( {{\alpha _{sc}}} \right)}}} \Big( {C\left( {{\alpha _{sc}}} \right) - \frac{d}{m}} \Big) \label{12a} \\
     & \mbox{s.t.} \quad \varepsilon \left( {{\alpha _{cc}},m} \right) \le {\varepsilon _c},\label{12b} \\
     & \quad \varepsilon \left( {{\alpha _{cs}},m} \right) + \left( {1 - \varepsilon \left( {{\alpha _{cs}},m} \right)} \right)\varepsilon \left( {{\alpha _{ss}},m} \right) \le {\varepsilon_s},\label{12c} \\
     & \quad\quad \frac{{\left( {1 - {a_c}} \right)P{{\left| {{\bf{v}}_c^H{{\bf{h}}_c}} \right|}^2}}}{{\sigma _c^2}} \le {\alpha _{sc}},\label{12d}   \\
     & \quad\quad \frac{{{a_c}P{\beta _i}}}{{\left( {1 - {a_c}} \right)P{\beta _i} + \sigma _i^2}} \ge {\alpha _{ci}},\label{12e}\\
     & \quad\quad \frac{{\left( {1 - {a_c}} \right)P{\beta _s}}}{{\sigma _s^2}} \ge {\alpha _{ss}}, \label{12f}\\
     & \quad\quad {\left| {{\bf{v}}_i^H{{\bf{h}}_i}} \right|^2} \ge {\beta _i}, \label{12g}\\
     & \quad\quad T \le T_{max}, m\ge0, \label{12h}\\
     & \quad\quad (\text{10f}). \label{12i}
  \end{align}
\end{subequations}
It is observed that (\ref{12a}), (\ref{12b}), (\ref{12c}) and (\ref{12g}) are non-convex. Define $\omega\!\left( {\alpha ,m} \right)\! = \!\sqrt {\frac{m}{{V\left( \alpha  \right)}}} \left( {C\left( \alpha  \right) - \frac{d}{m}} \right)\ln 2$, which is proved to be jointly convex in $\alpha $ and $m$ in our regime of interest in Lemma 3. Since ${\varepsilon _c}$ is much less than 0.5, we have $\omega_{cc}({\alpha _{cc}},m) > 0$, in which regime $Q\left( \omega_{cc}  \right)$ is convex and monotonically decreasing. Therefore, $\varepsilon({\alpha _{cc}},m)$ is jointly convex in $\alpha_{cc} $ and $m$. Besides, (\ref{12c}) is also convex in the considered regime according to Lemma 3.
\begin{lemma}
    \label{lemma3}
    $\omega \left( {\alpha ,m} \right) = \sqrt {\frac{m}{{V\left( \alpha  \right)}}} \left( {C\left( \alpha  \right) - \frac{d}{m}} \right)\ln 2$ is jointly convex in $\alpha$ and $m$ if the following condition holds:
    \begin{equation}
    \label{13}
    \begin{split}
    r \ge \frac{{ - {\Delta _b} + \sqrt {\Delta _b^2 - 4{\Delta _a}{\Delta _c}} }}{{2{\Delta _a}}},
    \end{split}
\end{equation}
where ${\Delta _a} = \frac{{8 + 9t}}{{4{t^2}}}$, ${\Delta _b} = \frac{{t\left( {6t + 8} \right) - \left( {3t + 8} \right)C\ln 2}}{{4{t^2}\ln 2}}$, ${\Delta _c} = \frac{{tC\ln 2\left( {4 - 3\ln 2} \right) + {t^2}\left( {C\ln 2 - 1} \right) - 4{C^2}{{\left( {\ln 2} \right)}^2}}}{{4{t^2}{{\left( {\ln 2} \right)}^2}}}$ with $t = {\alpha ^2} + 2\alpha$. 
\end{lemma}
\begin{proof}
    See appendix C in \cite{ref13}.
\end{proof}
\begin{lemma}
    \label{lemma4}
    $\varepsilon \left( {{\alpha _{cs}},m} \right) + \left( {1 - \varepsilon \left( {{\alpha _{cs}},m} \right)} \right)\varepsilon \left( {{\alpha _{ss}},m} \right)$ is jointly convex in ${\alpha _{cs}}$, ${\alpha _{ss}}$ and $m$ if (\ref{13}) in lemma 3 holds.
\end{lemma}

\begin{proof}
    See appendix B.
\end{proof}

As for (\ref{12g}), MM is applied to deal with it in an iterative manner. Specifically, the first-order Taylor expansion of $f\left( {{{\bf{v}}_i}} \right) = {\left| {{\bf{v}}_i^H{{\bf{h}}_i}} \right|^2}$ around the obtained optimal point ${\bf{v}}_i^{\left( q \right)}$ in the $q$th iteration can be derived as
\begin{equation}
    \label{14}
    \begin{split}
    & f\left( {{{\bf{v}}_i}} \right) = {\left| {{\bf{v}}_i^H{{\bf{h}}_i}} \right|^2} \ge \tilde f\Big( {{{\bf{v}}_i},{\bf{v}}_i^{\left( q \right)}} \Big) \\
    & = 2{\mathop{\rm Re}\nolimits} \Big( {{\bf{v}}{{_i^{\left( q \right)}}^H}{{\bf{h}}_i}{\bf{h}}_i^H{{\bf{v}}_i}} \Big) - {\mathop{\rm Re}\nolimits} \Big( {{\bf{v}}{{_i^{\left( q \right)}}^H}{{\bf{h}}_i}{\bf{h}}_i^H{\bf{v}}_i^{\left( q \right)}} \Big),  
    \end{split}
\end{equation}
which is a lower bound of $f\left( {{{\bf{v}}_i}} \right)$. Similarly, for the objective function $g\left( {{\alpha _{sc}},m} \right) = \sqrt {\frac{m}{{V\left( {{\alpha _{sc}}} \right)}}} \left( {C\left( {{\alpha _{sc}}} \right) - \frac{d}{m}} \right)$, the first-order Taylor expansion around the obtained point $\left( {\alpha _{sc}^{\left( q \right)},{m^{\left( q \right)}}} \right)$ in the $q$th iteration can be derived as
    \begin{align}
&g\left( {{\alpha _{sc}},m} \right) = \sqrt {\frac{m}{{V\left( {{\alpha _{sc}}} \right)}}} \left( {C\left( {{\alpha _{sc}}} \right) - \frac{d}{m}} \right) \nonumber\\
&\le \tilde g\left( {{\alpha _{sc}},m,\alpha _{sc}^{\left( q \right)},{m^{\left( q \right)}}} \right)\nonumber\\
&= \sqrt {\frac{{{m^{\left( q \right)}}}}{{V\big( {\alpha _{sc}^{\left( q \right)}} \big)}}} \left( {C\left( {\alpha _{sc}^{\left( q \right)}} \right) - \frac{d}{{{m^{\left( q \right)}}}}} \right) + \frac{1}{2}\left( {m - {m^{\left( q \right)}}} \right)\nonumber\\
 &{\left( {{m^{\left( q \right)}}V\left( {\alpha _{sc}^{\left( q \right)}} \right)} \right)^{ - \frac{1}{2}}} \left( {C\left( {\alpha _{sc}^{\left( q \right)}} \right) + \frac{d}{{{m^{\left( q \right)}}}}} \right) +\nonumber\\
 &\left( {{\alpha _{sc}} - \alpha _{sc}^{\left( q \right)}} \right)\sqrt {\frac{{{m^{\left( q \right)}}}}{{V\big( {\alpha _{sc}^{\left( q \right)}} \big)}}} \frac{1}{{1 + \alpha _{sc}^{\left( q \right)}}}\nonumber\\
 &\Bigg( {\frac{1}{{\alpha {{_{sc}^{\left( q \right)}}^2} + 2\alpha _{sc}^{\left( q \right)}}}\left( {\frac{d}{{{m^{\left( q \right)}}}} - C\left( {\alpha _{sc}^{\left( q \right)}} \right)} \right) + \frac{1}{{\ln 2}}} \Bigg),
\end{align}
which is an upper bound of $g\left( {{\alpha _{sc}},m} \right)$. Then the optimization problem in the $(q+1)$th iteration can be reformulated as
\begin{subequations}
  \begin{align}
     & {\rm{P4}}: \mathop {\min }\limits_{\scriptstyle{{\bf{v}}_i},m,{\alpha _{ci}}\hfill\atop
\scriptstyle,{\alpha _{si}},{\beta _i}\hfill} \tilde g\left( {{\alpha _{sc}},m,\alpha _{sc}^{\left( q \right)},{m^{\left( q \right)}}} \right) \label{16a} \\
     & \mbox{s.t.} \quad \tilde f\left( {{{\bf{v}}_i},{\bf{v}}_i^{\left( q \right)}} \right) \ge {\beta _i},\label{16b} \\
     & \quad \quad (\text{12b})-(\text{12f}), (\text{12h}), (\text{12i}). \label{16c}
  \end{align}
\end{subequations}
At this point, we have transformed P3 into a sequence of convex problems. In what follows, we will focus on the subproblem of ${a_c}$.

\subsubsection{Optimizing ${a_c}$ with fixed ${{\bf{v}}_i}$ and $m$} By substituting the results obtained in \textit{1)} into P2 and discarding the terms independent of ${a_c}$, the subproblem of ${a_c}$ can be written as
\begin{subequations}
  \begin{align}
     & {\rm{P5}}: \mathop {\min }\limits_{{a_c}} \sqrt {\frac{m}{{V\left( {{\gamma _{sc}}} \right)}}} \left( {C\left( {{\gamma _{sc}}} \right) - \frac{d}{m}} \right) \label{17a} \\
     & \mbox{s.t.} \quad (\text{10b}), (\text{10c}), (\text{11b}), (\text{11c}) \label{17b}.
  \end{align}
\end{subequations}
The objective function (\ref{17a}) and ${\varepsilon _{cc}}$ in (\ref{10b}) decrease monotonically with ${a_c}$. Therefore, P5 is equivalent to finding the maximal ${a_c^{(p+1)}}$ within $\left[{a_c^{(p)}},1\right)$ under the constraint of (\ref{10c}), where $p$ denotes the number of iterations of AO. Since the optimization variable is a scalar, the optimal solution can be found through one-dimension search. Accordingly, the overall algorithm is summarized in Algorithm 1.
\begin{figure}[!t]
  \begin{algorithm}[H]
    \caption{Security optimization based on AO \& MM }
    \label{alg1}
    \begin{algorithmic}[1]
      \STATE \textbf{Input} $P, {\sigma _i}, \rho, {\alpha _1}, {\alpha _2}, K, {c}, {d}, {\varepsilon _i}, {m_0}, {\zeta _1}, {\zeta _2}, {N_v}, {N_h}$,\\
      and $Iter_{\max }$;
      \STATE \textbf{Initialize} $a_c^{\left( 0 \right)}, {\bf{v}}_i^{\left( 0 \right)}, {m^{(0)}}, \varepsilon _{sc}^{(0)}, {\bf{\tilde v}}_i^{\left( 0 \right)}, {\tilde m^{(0)}}$ and $\tilde \varepsilon _{sc}^{(0)}$;
      \FOR {$p=1:Iter_{\max }$}
      \FOR{$q=1:Iter_{\max}$}
      \STATE Obtain ${\bf{\tilde v}}_i^{\left( q \right)}$, ${\tilde m^{(q)}}$ and $\tilde \varepsilon _{sc}^{(q)}$ for given $a_c^{\left( {p - 1} \right)}$, ${\bf{\tilde v}}_i^{\left( {q - 1} \right)}$ and ${\tilde m^{(q - 1)}}$;
      \IF {$\left| {\tilde \varepsilon _{sc}^{(q)} - \tilde \varepsilon _{sc}^{(q - 1)}} \right| \le {\zeta _1}$}
      \STATE Set ${\bf{v}}_i^{\left( p \right)} = {\bf{\tilde v}}_i^{\left( q \right)},{m^{(p)}} = {\tilde m^{(q)}}$;
      \STATE Break;
      \ENDIF
      \ENDFOR
      \STATE Obtain $a_c^{\left( p \right)}$ and $\varepsilon _{sc}^{(p)}$ for given ${\bf{v}}_i^{\left( p \right)}$ and ${m^{(p)}}$;
      \IF{$\left| {\varepsilon _{sc}^{(p)} - \varepsilon _{sc}^{(p - 1)}} \right| < {\zeta _2}$}
      \STATE Set $a_c^* = a_c^{\left( p \right)}$, ${\bf{v}}_i^* = {\bf{v}}_i^{\left( p \right)}$, $m* = {m^{(p)}}$, $\varepsilon _{sc}^* = \varepsilon _{sc}^{(p)}$;
      \STATE Break;
      \ENDIF
      \ENDFOR
      \STATE Set $m^{*}=\lceil m^{*}\rceil$;
      \STATE \textbf{Output} $a_c^*,{\bf{v}}_i^*,m*,\varepsilon _{sc}^*$.
    \end{algorithmic}
  \end{algorithm}
  \vspace{-0.8cm}
\end{figure}

\subsubsection{Complexity Analysis}
The subproblem of ${{\bf{v}}_i}$ and $m$ is implemented by the interior-point method with a complexity of $\mathcal{O}(k_1k_2(4N+4)^{3.5}\mathrm{log}(1/\epsilon))$, where $k_1$ and $k_2$ are the number of iterations for AO and MM, respectively. Moreover, $4N+4$ is the number of optimization variables and $\epsilon$ the accuracy of the interior-point method. Besides, the complexity of P5 is $\mathcal{O}(k_1\mathrm{log}(\psi/\psi_{th}))$, where where $\psi$ and $\psi_{th}$ denote the step-length factor and step-length threshold.

\section{Illustrative Results}
In this section, simulations are carried out to evaluate the performance of the proposed scheme. Unless otherwise stated, the parameters are set as: ${{\bf{c}}_S}={\left[ {0,0,10} \right]^T}$m, ${{\bf{c}}_R}={\left[ {35,20,10} \right]^T}$m, ${{\bf{c}}_{{U_c}}}={\left[ {40,0,0} \right]^T}$m, ${\bf{c}}_{{U_s}}={\left[ {40,40,0} \right]^T}$m, $\rho=-30$ dB, ${\alpha _1}={\alpha _2}=2.5$, $K=3$ dBm, $B=1.4$ MHz, $N=16$, $\sigma _c^2=\sigma _s^2=-80$ dBm, ${P_{\max }}=30$ dBm, ${\varepsilon _c}, {\varepsilon _s}={10^{ - 3}}$, ${T_{\max }}=0.715$ ms, ${m_{max}}=1000$ channel uses, ${d_c}={d_s}=100$, ${\zeta _1}={10^{ - 4}}$, ${\zeta _2}={10^{ - 10}}$, $Ite{r_{\max }}=30$.

\begin{figure}[t]
	\centering  
 	\subfigbottomskip=-1.5pt 
	\subfigcapskip=-1.5pt 
	\subfigure[Optimized ${\delta _{sc}}$]{
		\includegraphics[width=0.48\linewidth]{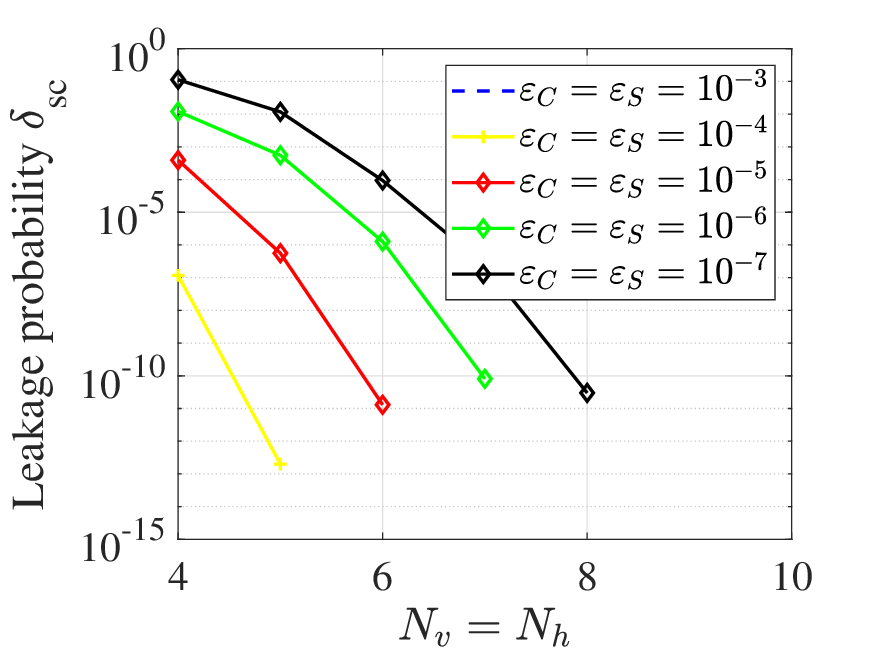}}
	\subfigure[Optimized $T$]{
		\includegraphics[width=0.48\linewidth]{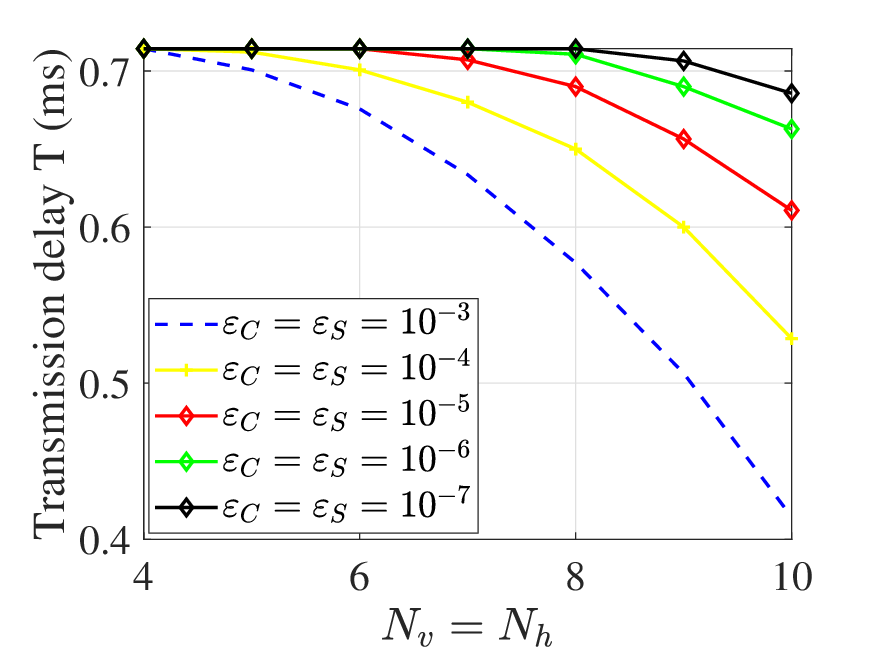}}
	\caption{Optimized ${\delta _{sc}}$ and $T$ versus the number of STAR-RIS elements}
\end{figure}

Fig. 3 illustrates the optimized ${\delta _{sc}}$ and $T$ by the proposed scheme versus the number of STAR-RIS elements under different reliability constraints. Compared with the non-optimized case in Fig. 2, a low level of $\delta_{sc}$ can be achieved while satisfying the strict reliability constraints of both users. At the same time, even though the upper limit of transmission delay has been set to $T_{max}=0.715$ ms (i.e., $m_{max}=1000$), the optimized result can achieve an even lower level. The reason is that the restrictive trade-off between ${\varepsilon _{cc}}$, ${\varepsilon _{ss}}$, ${\delta _{sc}}$ and $T$ is improved by the additional freedom degrees brought by STAR-RIS. Specifically, by constructively/destructively reconfiguring the channel between S and $U_s$/$U_c$, the required $a_s$ and $m$ to achieve the reliability requirement of $U_s$ can be reduced, which also enables a lower leakage probability. At the same time, to meet the reliability requirement of $U_c$, more power can be allocated to $s_c$ to compensate for the destructed channel between S and $U_c$ and the reduced blocklength. In other words, through joint optimization, the resources including beamforming, allocated power and blocklength available to ${U_c}$ for decoding $s_s$ can be reduced for a lower ${\delta _{sc}}$ without violating the constraints of ${\varepsilon _{cc}}$ and ${\varepsilon _{ss}}$. It is also seen that the increase in STAR-RIS elements number brings lower ${\delta _{sc}}$ and $T$ due to a higher degree of freedom. Besides, ${\delta _{sc}}$ and $T$ tend to increase as the preset thresholds ${\varepsilon _c}$ and ${\varepsilon _s}$ decrease due to the restrictive interplay between ${\varepsilon _{cc}}$, ${\varepsilon _{ss}}$ ${\delta _{sc}}$ and $T$.

\begin{figure}[!t]
  \centering
  \includegraphics[width=0.48\linewidth]{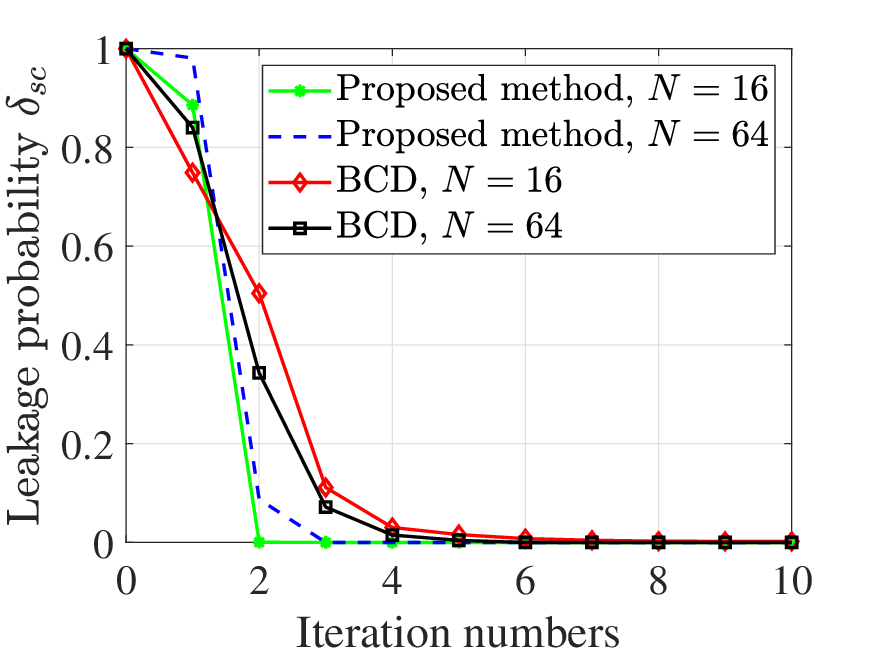}
  \caption{Convergence performance of the proposed algorithm}
  \label{fig4}
\end{figure}

Fig. 4 demonstrates the convergence performance of the proposed algorithm with different numbers of STAR-RIS elements. The block coordinate descent (BCD)-based method is conducted as the benchmark \cite{ref21}, where three subproblems of beamforming, power allocation and blocklength are iteratively solved via successive convex approximation and exhaustive search. It is illustrated that the proposed scheme shows a faster convergence speed  than the benchmark due to the exploration of the joint convexity of the optimization variables, which indicates a higher adaptability for latency-sensitive scenarios.

\section{Conclusion}
In this article, a STAR-RIS-aided untrusted NOMA-SPC system was considered to enhance the security performance while maintaining the required xURLLC performances. The reliability and security performance were characterized via probabilistic metrics for a better demonstration of the trade-off among reliability, security and latency. Through the assistance of STAR-RIS and joint design of passive beamforming, power allocation and blocklength, the trade-off can be greatly improved and the security can be significantly improved while fulfilling the customized reliability and latency requirements.

\begin{appendices}
    \section{Proof of lemma 1}
        Assume “$<$” in (\ref{10b}) holds at the optimal solution. Since $\frac{{\partial {\varepsilon _{cc}}}}{{\partial {{\left| {{\bf{v}}_c^H{{\bf{h}}_c}} \right|}^2}}} \le 0$ and $\frac{{\partial {\delta _{sc}}}}{{\partial {{\left| {{\bf{v}}_c^H{{\bf{h}}_c}} \right|}^2}}} \ge 0$, lower ${\delta _{sc}}$ can always be found through ${\left| {{\bf{v}}_c^H{{\bf{h}}_c}} \right|^2} < {\left| {{\bf{v}}_c^{*H}{{\bf{h}}_c}} \right|^2}$ without degrading ${\varepsilon _{ss}}$. Therefore, ${\varepsilon _{cc}} = {\varepsilon _c}$ holds at optimal solution, and (\ref{10a}) can be reduced to ${\bar \varepsilon _{sc}}$. Similarly, assume the optimal solution satisfies $a_c^* + a_s^* < 1$. Since $\frac{{\partial {{\bar \varepsilon }_{sc}}}}{{\partial {{\left| {{\bf{v}}_c^H{{\bf{h}}_c}} \right|}^2}}} \le 0$, lower ${\bar \varepsilon _{sc}}$ can always be achieved through a higher ${a_c}$ and a lower ${\left| {{\bf{v}}_c^H{{\bf{h}}_c}} \right|^2}$ satisfying ${\varepsilon _{cc}} = {\varepsilon _c}$. Meanwhile, ${\varepsilon _{ss}}$ can be reduced due to lower ${\varepsilon _{cs}}$. Therefore, $a_c^* + a_s^* = 1$ holds.

    \section{Proof of lemma 3}
        The Hessian matrix of ${\varepsilon _{ss}}\left( {{\omega _{cs}},{\omega _{ss}}} \right) = Q\left( {{\omega _{cs}}} \right) + \left( {1 - Q\left( {{\omega _{cs}}} \right)} \right)Q\left( {{\omega _{ss}}} \right)$ is represented as follows:
        \begin{equation}
             \label{18}
             \begin{split}
             {\bf{H}} = \left( {\begin{array}{*{20}{c}}
{\frac{{{\partial ^2}{\varepsilon _{ss}}}}{{\partial {\omega _{cs}}^2}}}&{\frac{{{\partial ^2}{\varepsilon _{ss}}}}{{\partial {\omega _{cs}}\partial {\omega _{ss}}}}}\\
{\frac{{{\partial ^2}{\varepsilon _{ss}}}}{{\partial {\omega _{ss}}\partial {\omega _{cs}}}}}&{\frac{{{\partial ^2}{\varepsilon _{ss}}}}{{\partial {\omega _{ss}}^2}}}
\end{array}} \right).
  \end{split}
\end{equation}
The determinant of the matrix can be derived as:
  \begin{align}
  \label{19}
    & \det \left( {\bf{H}} \right) = \frac{1}{2\pi}{e^{ - \frac{{\omega _{cs}^2}}{2} - \frac{{\omega _{ss}^2}}{2}}}\big( {1 - Q\left( {{\omega _{cs}}} \right)} \big)\big( {1 - Q\left( {{\omega _{ss}}} \right)} \big) \nonumber \\
& \qquad \qquad \qquad \qquad \qquad \quad {{\omega _{cs}}{\omega _{ss}}} - \frac{1}{{4{\pi ^2}}}{e^{ - \omega _{cs}^2 - \omega _{ss}^2}} \nonumber \\
 & = \frac{1}{{2\pi }}{e^{ - \frac{{\omega _{cs}^2}}{2} - \frac{{\omega _{ss}^2}}{2}}}\Big( {\left( {1 - {\varepsilon _{ss}}} \right){\omega _{cs}}{\omega _{ss}} - \frac{1}{{2\pi }}{e^{ - \frac{{\omega _{cs}^2}}{2} - \frac{{\omega _{ss}^2}}{2}}}} \Big) \nonumber \\
& \mathop  \ge \limits^a \frac{1}{{2\pi }}{e^{ - \frac{{\omega _{cs}^2}}{2} - \frac{{\omega _{ss}^2}}{2}}}\Big( \left( {1 - 0.3} \right){{\big(Q^{-1}(0.3)\big)}^2} -  \nonumber \\
& \qquad \qquad \qquad \qquad \qquad \quad \frac{1}{{2\pi }}{e^{ - {{\big(Q^{-1}(0.3)\big)}^2}}} \Big) \ge 0,  
  \end{align}
  where “a” holds since ${\varepsilon _{ss}} \!\ll\! 0.3$, $Q\left( {{\omega _{cs}}} \right), Q\left( {{\omega _{ss}}} \right) \!<\! {\varepsilon _{ss}}$, i.e., ${\omega _{cs}}, {\omega _{ss}} > Q^{-1}(0.3)$, and $h({\varepsilon _{ss}}, \omega _{cs}, \omega _{ss})=\left( {1 - {\varepsilon _{ss}}} \right){\omega _{cs}}{\omega _{ss}}- \frac{1}{{2\pi }}{e^{ - \frac{{\omega _{cs}^2}}{2} - \frac{{\omega _{ss}^2}}{2}}}$ decreases monotonically with ${\varepsilon _{ss}}$ while increases with $\omega _{cs}$ and $\omega _{ss}$. Thus ${\varepsilon _{ss}}\left( {{\omega _{cs}},{\omega _{ss}}} \right)$ is jointly convex in ${\omega _{cs}}$ and ${\omega _{ss}}$ in our interested regime. Moreover, recall that ${\omega _{cs}}\left( {{\alpha _{cs}},m} \right)$ and ${\omega _{ss}}\left( {{\alpha _{ss}},m} \right)$ are concave in this regime, and ${\varepsilon _{ss}}\left( {{\omega _{cs}},{\omega _{ss}}} \right)$ decreases monotonically in ${\omega _{cs}}$ and ${\omega _{ss}}$. Therefore, $\varepsilon \left( {{\alpha _{cs}},m} \right) + \left( {1 - \varepsilon \left( {{\alpha _{cs}},m} \right)} \right)\varepsilon \left( {{\alpha _{ss}},m} \right)$ is jointly convex in ${\alpha _{cs}}$, ${\alpha _{ss}}$ and $m$.
\end{appendices}


\begin{thebibliography}{00}

\bibitem{ref1}
X. Zhang, D. Zhang, B. Shim, G. Han, D. Zhang and T. Sato, "Sparse Superimposed Coding for Short-Packet URLLC," \textit{IEEE Internet Things J.}, vol. 9, no. 7, pp. 5275-5289, Apr. 2022.


\bibitem{ref2}
M.2160 : Framework and overall objectives of the future development of IMT for 2030 and beyond, \textit{https://www.itu.int/en/ITU-R/study-groups/rsg5/rwp5d/imt-2030/Pages/default.aspx}, page 14-17.

\bibitem{ref3}
Y. Polyanskiy, H. V. Poor, and S. Verdu, “Channel Coding Rate in the FBL Regime,” \textit{IEEE Trans. Inf. Theory}, vol. 56, no. 5, pp. 2307–2359, May 2010.
  
\bibitem{ref4}
W. Yang, R. F. Schaefer, and H. V. Poor, “Wiretap Channels: Nonasymptotic Fundamental Limits,” \textit{IEEE Trans. Inf. Theory}, vol. 65, no. 7, pp. 4069–4093, Jul. 2019.

\bibitem{ref5}
T. -H. Vu, T. -V. Nguyen, T. -T. Nguyen and S. Kim, "Performance Analysis and Deep Learning Design of Wireless Powered Cognitive NOMA IoT Short-Packet Communications With Imperfect CSI and SIC," \textit{IEEE Internet Things J.}, vol. 9, no. 13, pp. 10464-10479, Jul. 2022.

\bibitem{ref6}
Z. Ding, X. Lei, G. K. Karagiannidis, R. Schober, J. Yuan and V. K. Bhargava, "A Survey on Non-Orthogonal Multiple Access for 5G Networks: Research Challenges and Future Trends," \textit{IEEE J. Sel. Areas Commun.}, vol. 35, no. 10, pp. 2181-2195, Oct. 2017.

\bibitem{ref7}
D. -D. Tran, S. K. Sharma, S. Chatzinotas, I. Woungang and B. Ottersten, "Short-Packet Communications for MIMO NOMA Systems Over Nakagami-m Fading: BLER and Minimum Blocklength Analysis," \textit{EEE Trans. Veh. Technol.}, vol. 70, no. 4, pp. 3583-3598, Apr. 2021.

\bibitem{ref8}
X. Zhang, L. Yang, Z. Ding, J. Song, Y. Zhai and D. Zhang, "Sparse Vector Coding-Based Multi-Carrier NOMA for In-Home Health Networks," in IEEE Journal on Selected Areas in Communications, vol. 39, no. 2, pp. 325-337, Feb. 2021.

\bibitem{ref91}
S. Jia et al., "Secrecy Performance Analysis of UAV-Assisted Ambient Backscatter Communications With Jamming," \textit{IEEE Trans. Wireless Commun.}, vol. 23, no. 12, pp. 18111-18125, Dec. 2024.


\bibitem{ref9}
L. Guo, J. Jia, J. Chen and X. Wang, "Secure Communication Optimization in NOMA Systems With UAV-Mounted STAR-RIS," \textit{IEEE Trans. Inf. Forensics Secur.}, vol. 19, pp. 2300-2314, 2024.

\bibitem{ref10}
Z. Xiang, W. Yang, Y. Cai, J. Xiong, Z. Ding and Y. Song, "Secure Transmission in a NOMA-Assisted IoT Network With Diversified Communication Requirements," \textit{IEEE Internet Things J.}, vol. 7, no. 11, pp. 11157-11169, Nov. 2020.

\bibitem{ref11}
X. Lai, T. Wu, Q. Zhang and J. Qin, "Average Secure BLER Analysis of NOMA Downlink Short-Packet Communication Systems in Flat Rayleigh Fading Channels," \textit{IEEE Trans. Wirel. Commun.}, vol. 20, no. 5, pp. 2948-2960, May 2021.

\bibitem{ref12}
Z. Feng, H. Lu, N. Zhao, Z. Shi, Y. Chen and X. Wang, "Secure Transmission of UAV Control Information via NOMA," \textit{IEEE Trans. Commun.}, vol. 72, no. 8, pp. 4648-4660, Aug. 2024.

\bibitem{ref13}
Y. Zhu, X. Yuan, Y. Hu, R. F. Schaefer and A. Schmeink, "Trade Reliability for Security: Leakage-Failure Probability Minimization for Machine-Type Communications in URLLC," \textit{IEEE J. Sel. Areas Commun.}, vol. 41, no. 7, pp. 2123-2137, Jul. 2023.

\bibitem{ref14}
T. -H. Vu, T. -V. Nguyen, Q. -V. Pham, D. Benevides da Costa and S. Kim, "STAR-RIS-Enabled Short-Packet NOMA Systems," \textit{EEE Trans. Veh. Technol.}, vol. 72, no. 10, pp. 13764-13769, Oct. 2023.  

\bibitem{ref15}
S. Lv, X. Xu, S. Han and P. Zhang, "RIS-Enhanced Secure Transmission in MTC Networks With Finite Blocklength," \textit{IEEE Trans. Commun.}, vol. 71, no. 6, pp. 3513-3527, Jun. 2023.    


    

\bibitem{ref16}
Y. Yang and L. Hanzo, "Permutation-Based Short-Packet Transmissions Improve Secure URLLCs in the Internet of Things," \textit{IEEE Internet Things J.}, vol. 10, no. 12, pp. 11024-11037, Jun. 2023.

\bibitem{ref17}
M. Ahmed et al., "A Survey on STAR-RIS: Use Cases, Recent Advances, and Future Research Challenges," \textit{IEEE Internet Things J.}, vol. 10, no. 16, pp. 14689-14711, Aug. 2023.

\bibitem{ref18}
C. Wu, C. You, Y. Liu, X. Gu and Y. Cai, "Channel Estimation for STAR-RIS-Aided Wireless Communication," in IEEE Communications Letters, vol. 26, no. 3, pp. 652-656, Mar. 2022.

\bibitem{ref19}
C. Wu, C. You, Y. Liu, S. Han and M. D. Renzo, "Two-Timescale Design for STAR-RIS-Aided NOMA Systems," \textit{IEEE Trans. Commun.}, vol. 72, no. 1, pp. 585-600, Jan. 2024.

\bibitem{ref20}
M. Katwe, K. Singh, B. Clerckx and C. -P. Li, "Improved Spectral Efficiency in STAR-RIS Aided Uplink Communication Using Rate Splitting Multiple Access," \textit{IEEE Trans. Wireless Commun.}, vol. 22, no. 8, pp. 5365-5382, Aug. 2023.

\bibitem{ref21}
Y. Lou, Y. Zou, H. Wang, L. Zhai and Y. Li, "Joint Phase Shifts and Blocklength Resources Optimization in RIS-Assisted NOMA Short-Packet Systems," \textit{IEEE Trans. Veh. Technol.}, early access.
\end{thebibliography}
\end{document}